\documentclass[twocolumn, showkeys]{IEEEtran}

\usepackage{epsfig}
\usepackage{graphicx}
\usepackage{url}

\usepackage{color} 

\usepackage{authblk}
\usepackage{amssymb,amsmath,amsthm}       %
\usepackage{amsfonts}      %

\usepackage{mathtools}     %
\usepackage{float}

\def\BibTeX{{\rm B\kern-.05em{\sc i\kern-.025em b}\kern-.08em T\kern-.1667em\lower.7ex\hbox{E}\kern-.125emX}}

\newtheorem{theorem}{Theorem}
\newtheorem{lemma}{Lemma}

\begin{document}

\title{Physical geometry of channel degradation}

\author{Steve Huntsman}
\affil{sch213@nyu.edu\vspace*{-1cm}}

\date{\today}

\IEEEoverridecommandlockouts
\IEEEpubid{\makebox[\columnwidth]{\copyright2019 BAE Systems, Inc. \hfill} \hspace{\columnsep}\makebox[\columnwidth]{ }}

\maketitle

\thispagestyle{empty}

\begin{abstract}
We outline a geometrical correspondence between capacity and effective free energy minima of discrete memoryless channels. This correspondence informs the behavior of a timescale that is important in effective statistical physics.
\end{abstract}


%

\section{\label{sec:tcc}Thermodynamics and channel capacity}

It seems inevitable that the second law of thermodynamics and the propensity for noise to degrade communications channels should be related. In particular, the principle of minimum free energy is a natural candidate for providing a physical interpretation of channel degradation. Specifically, since free energy measures the energy available to perform work, and channel capacity measures the information available for transmission \cite{Ash}, it is natural to anticipate an analogy with something like a principle of minimum channel capacity. 

However, while attempts to link channel capacity and something like a manifestly physical (versus, e.g., a variational/Bethe) free energy have been made \cite{Merhav,Ford3,KarnaniPA,ChiangBoyd}, none of these attempts are fully satisfactory. In \S \ref{sec:mife}, we identify shortcomings along these lines in \cite{Merhav,Ford3}. Meanwhile, \cite{KarnaniPA} avoids detailed discussion of channel capacity altogether, and in \cite{ChiangBoyd}, the ``temperature'' is fixed to unity, so ``energy'' levels are really just conditional entropies. 

More generally, entropy is usually the primary material used for building bridges between information theory and statistical physics, and physically ubiquitous terms of the form [energy]/[temperature] are almost invariably treated as inseparable except perhaps formally (see, e.g., \cite{PelegESK}). This tactic cloaks essentially information-theoretical and mathematical observations involving probabilities in the guise of physics (much as with, e.g., the transient fluctuation theorem \cite{JQQ}).

In contrast, we aim to place temperature and energy on conceptual ground similar to that which entropy has long occupied, albeit using the additional datum of a characteristic timescale. This more expressly physical perspective connects information theory and statistical physics via the geometry of free energy and channel capacity landscapes.

Our basic result is that for the a discrete memoryless channel, the effective free energy is correlated with the channel capacity in a very particular way, so that minima of effective free energy occur near minima of channel capacity. By way of example, we show that the same is not true of entropy. Our results are both consistent with and complementary to \cite{Reiss,ReissHuang,QianR} and more generally provide support for purely physical aspects of our proposed framework based on an effective temperature \cite{Huntsman3}.


\section{\label{sec:efft}Effective statistical physics of finite stationary systems}

We recall the form of an effective temperature $\beta^{-1}$ for finite stationary systems, following \cite{Huntsman3} and omitting details. Consider a statistically stationary system with $n < \infty$ states and characterized by a nondegenerate probability distribution $p := (p_1,\dots,p_n) > 0$ along with a timescale $t_\infty$. 

\begin{theorem} \cite{Huntsman3} Up to an overall constant $\hbar$ with units of action ($ = [\text{time}] \cdot [\text{energy}]$), there is a unique bijection $t := t_\infty p \leftrightarrow (E, \beta^{-1})$ compatible with the Gibbs relation $Z^{-1}e^{-\beta E_k} = p_k$ and dimensional considerations. This bijection is determined by
\begin{equation}
\label{eq:temperature}
\beta(t) = t_{\infty} \lVert p \rVert_2 \cdot \sqrt{\lVert \gamma \rVert_2^2 +1},
\end{equation}
where $\gamma_k := \beta E_k = \frac{1}{n}\sum_{j=1}^{n} \log p_j - \log p_k$.
\end{theorem}

\begin{proof}[Proof (outline: for details, see \cite{Huntsman3})] Now $\|\beta \cdot (E,\beta^{-1}) \|_2^2 = \|\gamma\|_2^2 +1$, and since $\beta = \|\beta \cdot (E,\beta^{-1}) \|_2/\|(E,\beta^{-1})\|_2$, we only need to compute the denominator on the right to close the equations. Dimensional analysis (say, by a canonical transformation of a notional Hamiltonian that rescales time) turns out to require a relation of the form $\beta = \hbar^{-1} t_\infty \Psi(p)$. It can also be shown that $\|t\|_2 = \|t'\|_2 \Leftrightarrow \|(E,\beta^{-1})\|_2 = \|(E',\beta'^{-1})\|_2$, basically because $p$ is constant on rays through the origin in either $t$-coordinates or $(E,\beta^{-1})$-coordinates. From this, it follows that $u:= \|t\|_2 1 /\sqrt{n}$ corresponds to $E(u) = 0$, where here $1$ indicates the vector with all unit components. Now $$\beta^2 = \frac{\|\beta \cdot (E,\beta^{-1}) \|_2^2}{\|(E,\beta^{-1})\|_2^2} = \frac{\|\gamma\|_2^2 +1}{\beta(u)^{-2}}$$ so that $\beta = \beta(u) \cdot \sqrt{\lVert \gamma \rVert_2^2 +1}$. But by dimensional analysis and scaling behavior, $\beta(u) = \hbar^{-1} \|u\|_2 = \hbar^{-1} \|t\|_2 = \hbar^{-1} t_\infty \|p\|$. The result follows.
\end{proof}

\begin{figure}[htbp]
\centering
\includegraphics[trim = 69mm 117mm 70mm 124mm, clip, width=.8\columnwidth,keepaspectratio]{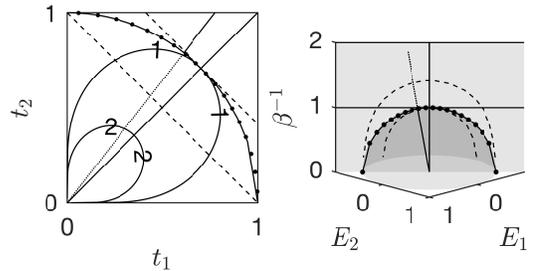}
\caption{ \label{fig:t2H} Geometry of the bijection $(p, t_\infty) \leftrightarrow (E, \beta^{-1})$ for $n = 2$ states. Level curves of $\beta^{-1} = 1,2$ (solid contours) and of $t_\infty = 1, \sqrt{2}$ (dashed contours) are shown in both coordinate systems. The bijection is also shown explicitly for circular arcs and rays.}
\end{figure} 

Applications of the theorem rest on identifying $t_\infty$. Though for a Markov process it can be shown (on the basis of intensivity) that $t_\infty$ cannot \emph{exactly} be the $L^2$ mixing time $t_{mix}$ described in \S \ref{sec:mixingTime}, it must still be broadly similar, and in practice we usually regard $t_{mix}$ as a proxy or avatar of $t_\infty$.
\footnote{In computations, we identify $t_{mix}$ and $t_\infty$ without further comment. For background on and algorithms for computing $t_{mix}$, see, \S \ref{sec:mixingTime} and \cite{LPW}.}

\section{\label{sec:dmcs}Discrete memoryless channels}

A discrete memoryless channel (DMC) is specified by a matrix
\footnote{
We only consider square channel matrices, but this is not that restrictive \cite{Cheng,Takano}. 
}
of conditional probabilities $W_{jk} = \mathbb{P}(y_k|x_j)$, where $x_j$ and $y_k$ respectively denote input and output symbols corresponding to realizations of random variables $X$ and $Y$. That is, $W_{jk}$ is the probability that if Alice transmits $x_j$, then Bob receives $y_k$. An input distribution $p_j = \mathbb{P}(x_j)$ yields joint probabilities $V_{jk} := \mathbb{P}(x_j,y_k) = p_j W_{jk}$. The channel coding theorem states that reliable communication is possible over a DMC iff the ratio of informative to transmitted bits is less than the channel capacity $C := \sup_p I(X;Y)$, where $I(X;Y)$ is the mutual information. 

The evident similarity of a DMC to a discrete-time Markov process should inform any attempt to treat the information theory of DMCs on the same footing as statistical physics. In particular, effective thermodynamical quantities for a DMC should bear some similarity to those for a Markov chain. The simplest approach to defining effective thermodynamical quantities that respects this guideline is to take the capacity-achieving (versus the invariant) distribution for $p$ and (some quantity very much like) the corresponding chain's $L^2$ mixing time $t_{mix} \equiv t_{mix}(W)$ for $t_\infty$. As we shall see below, this recipe is successful at providing a thermodynamical interpretation of channel degradation.

\subsection{\label{sec:mife}Helmholtz free energy and channel capacity }

\subsubsection{\label{sec:srefed}Similarity of relative entropy and a free energy difference}

Formally writing $p_j = \exp(-\beta E_j - \log Z)$ and writing $H$ and $F = -\beta^{-1} \log Z$ for entropy and free energy, respectively, it is easy to show that the relative entropy satisfies 
$D ( p \big \| \tilde p ) = -H(p)+ H(\tilde p) - \tilde \beta \sum_j (\tilde p_j - p_j) \tilde E_j$. In particular, if $\tilde \beta \sum_j p_j \tilde E_j \approx \beta \sum_j p_j E_j$,
then 
$D ( p \big \| \tilde p ) \approx \beta_* (F - \tilde F)$,
where $\beta \approx \beta_* \approx \tilde \beta$.
\footnote{NB. In \cite{Merhav}, these approximations are essentially treated as equalities.}
That is, the relative entropy is similar to a (Helmholtz) free energy difference, particularly when temperature factors are ignored. 

This similarity is exploited in many settings. For instance, $F$ and $\tilde F$ are frequently framed as the Helmholtz and variational/Bethe free energies, respectively: the latter is minimized as a proxy objective function in mean field theory and Bayesian estimation via the well-known EM algorithm \cite{MacKay}. 
The Blahut-Arimoto algorithm \cite{Blahut,Arimoto} for computing the capacity of a DMC fits into this context \cite{CsiszarTusnady}. It is therefore natural to expect that minimizing $\beta F = - \log Z$ is similar---though not equivalent---to minimizing $C$. 

As we shall see below, decomposing $\beta F$ into its constituent terms in a physically reasonable and meaningful way highlights the correspondence between $C$ and $F$. It also reconciles this correspondence with one between $C$ and a distinct notion of temperature for discrete \emph{noiseless} channels in \cite{Reiss}.
\footnote{
The essential observation here is that in the limit of large block length, the combination of a DMC and a capacity-approaching block code (e.g., a polar code \cite{SRDR}) essentially amount to a discrete noiseless channel.
}

\subsubsection{\label{sec:factor}Thermodynamic interpretation of factoring the joint distribution of a DMC}

The relationship between mutual information and free energy can be brought into clearer focus through a result of \cite{Ford3}.
\footnote{
Regarding the attempted interpretation in \cite{Ford3}, framing the quantity $F + \Delta W$ as a Gibbs free energy would impose specific requirements on the form of $\Delta W$ \emph{vis-\`a-vis} conjugate variables that apparently need not be met in general. However, there is nevertheless a useful analogy, cf. \cite{QianR} and the discussion at the end of \S \ref{sec:remarks}.
} 
Writing $V_{jk} := \mathbb{P}(x_j, y_k) = p_j W_{jk} \equiv \exp(-\beta E_{jk})/Z$, Alice and Bob's (marginal) energies are respectively $E^{(A)}_j = F - \frac{1}{\beta} \log \sum_k V_{jk}$ and $E^{(B)}_k = F - \frac{1}{\beta} \log \sum_j V_{jk}$, where as usual $Z = e^{-\beta F}$. 
\begin{theorem} \cite{Ford3}
Alice and Bob's mutual information is $I(X;Y) = -\beta (F + \Delta W)$
where $\Delta W = \sum_{j,k} V_{jk} (E_{jk} - [E^{(A)}_j + E^{(B)}_k])$ is the work associated with factoring the joint distribution of the DMC into its marginals.
\qed
\end{theorem}

Here the correspondences $\delta Q \equiv \sum_j dp_j \cdot E_j$ and $\delta W \equiv \sum_j p_j \cdot d E_j$ (for which see, e.g., (1.17-18) of \cite{Hill}) provide the usual basis for interpreting the thermodynamical concepts of heat and work within the context of statistical physics.

\subsubsection{\label{sec:remarks}Remarks}

In general, the mutual information and corresponding free energy difference are plainly not identical by arguments above, though it is nevertheless commonplace to draw (perhaps overly) suggestive links between these quantities. Yet it is also natural to expect the extrema of these two quantities to be more closely related. We will illustrate this dichotomy below through the lens of DMCs. 

While here we treat DMCs as (quasi-) physical systems, one can also treat some physical systems as DMCs. This tactic has been used to analyze kinetics of signal transduction in biological cells by highlighting a correlation between channel capacity and (Gibbs) free energy expenditure, in particular a precise identification of their respective minima \cite{QianR}. 

\section{\label{sec:phenomenology}Channel capacity and free energy}

\subsection{\label{sec:simple}Binary input/output DMCs and zero-capacity DMCs}

\begin{theorem}
For a binary input/output DMC (BIODMC),
\begin{equation}
\label{eq:argminsequal}
\arg \min C = \arg \min F_{mix},
\end{equation}
where $F_{mix}$ is the value of the effective free energy $F$ corresponding to the capacity-achieving distribution and using $t_\infty = t_{mix}$.
\end{theorem}

\begin{proof}[Proof (outline)] Both minima occur on the diagonal $W_{21} = 1-W_{12}$. A proof of this for $C$ amounts to showing that the capacity is zero only on the diagonal, and thus minimized precisely there. (This seems intuitively obvious, but a remarkably elaborate series of elementary calculations appears necessary to actually prove it, starting with deriving the capacity-maximizing distribution and then carefully establishing differentiability and the limiting value itself on the diagonal.) For $F_{mix} := -\beta_{mix}^{-1} \log Z$, a calculation shows that the $L^2$ mixing time $t_{mix}$ has its extremum on the diagonal, and another tedious calculation shows that $\beta_{mix} \cdot D_{W_\bullet} F_{mix}$ is zero there, which in turn yields that the derivative of $F_{mix}$ is zero. Mechanically checking the boundary for extrema and invoking a symmetry argument to establish that the extremum of $F_{mix}$ is indeed a minimum completes the proof.
\footnote{A detailed and brutally straightforward proof is at \url{https://bit.ly/33Wv1jc}.}
\end{proof}

Figure \ref{fig:BIODMC} illustrates the theorem and shows that \eqref{eq:argminsequal} continues to be true for a BIODMC if we hold any of the entries of the channel matrix fixed. In \S \ref{sec:nonsingulardmcs}, we consider the situation of a constrained channel matrix in greater generality. While for a BIODMC, $C$ and $F_{mix}$ have very similar behavior, we shall see that for generic DMCs the relationship is more complex than \eqref{eq:argminsequal}. Nevertheless, the result of minimizing $F_{mix}$ continues to be related to that of minimizing $C$. 

\begin{figure}[htbp]
\centering
\includegraphics[trim = 60mm 125mm 60mm 115mm, clip, width=\columnwidth,keepaspectratio]{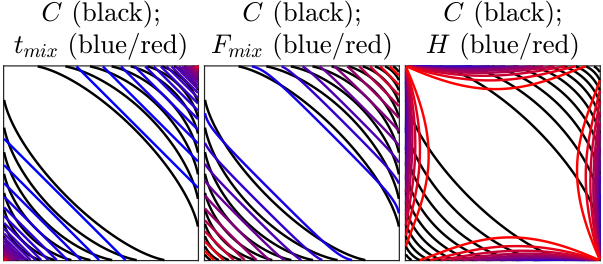}
\caption{ \label{fig:BIODMC} (Left) Contours of $C$ (black) and $t_{mix}$ (values increasing from {\color{blue}blue} to {\color{red}red}) of a BIODMC over the unit square with coordinates $(W_{12},W_{21})$. (Center; right) As in left panel, but for $F_{mix}$ and $H$, respectively.}
\end{figure} 

\begin{lemma}
For a DMC with square channel matrix, $C = 0 \iff t_{mix} = 1$.
\end{lemma}

\begin{proof}
The mutual information of two random variables is zero iff they are independent. It follows that a DMC with a square channel matrix $W$ has zero capacity iff $V_{jk} \equiv p_j q_k$, i.e. $W_{jk} \equiv q_k$, or in vector notation $W = 1q$, where $q$ is the row vector with $k$th entry $q_k$. A line of algebra shows that $qW = q$. Using the notation and results of \S \ref{sec:mixingTime}, a few short and elementary calculations suffice to show in turn that $W^\dagger = W$, $W^\dagger W = W$, $S = q^T q - \text{diag}(q)$, $T = S$, and $U = I - \sqrt{q}^T \sqrt{q}$, where the square root is componentwise. 

Now \eqref{eq:lambdastar4} gives that $t_{mix}^{-1} \equiv \lambda_* = \inf (\mbox{spec}(U) \backslash \{0\})$. In order to compute this, we note that
\begin{eqnarray}
\det(U - \lambda I) & = & \det((1-\lambda)I - \sqrt{q}^T \sqrt{q}) \nonumber \\
& = & (1 - \sqrt{q}(1-\lambda)^{-1}I \sqrt{q}^T) \cdot \det((1-\lambda)I) \nonumber \\
& = & (1 - (1-\lambda)^{-1})(1-\lambda)^n,
\end{eqnarray}
where the second equality follows from the matrix determinant lemma $\det(A + uv^T) = (1 + v^TA^{-1}u) \cdot \det(A)$ and the third equality follows from $\sum_\ell q_\ell = 1$. The determinant has a single root at $\lambda = 0$ and an $n$-fold root at $\lambda = 1$, so $\lambda_* = 1$ and we can conclude that for a DMC with a square channel matrix, $C = 0 \Rightarrow t_{mix} = 1$.

Meanwhile, we have in general that $t_{mix} \ge 1$ for a discrete-time Markov-like system such as a DMC, since $t_{mix} \equiv \lambda_*^{-1}$ and $\lambda_* = 1 - \sup_f \mbox{Var}_p(Wf)/\mbox{Var}_p(f) \le 1$ (cf. \eqref{eq:lambdastar3}). Suppose that in fact $t_{mix} = 1$, or equivalently that $\mbox{Var}_p(Wf) = 0$ for all $f$ and $p$. Then it must be the case that $Wf$ is always a constant (that depends on $f$), and this in turn requires $W \equiv 1q$. This yields the result.
\end{proof}


Combining the lemma with a continuity argument gives a reasonably general qualitative explanation of why the minima of $t_{mix}$ and $C$ are correlated. 


\subsection{\label{sec:nonsingulardmcs}Nonsingular constrained DMCs}

In this section, through a combination of numerical examples and partial analytical results, we illustrate the basic phenomenology relating capacity and free energy minimization for \emph{constrained} DMCs where only some of the channel matrix entries are allowed to vary. Specifically, we have the following general principle: if $F \approx \min F$, with degeneracies (such as in Figure \ref{fig:pardmc2}) removed by requiring $t_\infty \approx \min t_\infty$ (as in the lemma above), then $C \approx \min C$. The explanation hinges on the principle, presented in \S \ref{sec:basinsfollowcorners}, that ``basins of $F$ follow corners of $C$''. That is, as $W$ is subject to variable constraints, the corresponding minima of $F$ coalesce near the comparatively unconstrained minimum of $C$. 

While in the analysis below we leverage the so-called Muroga formula, there is no more general analytical expression known for the channel capacity of a DMC. In cases where the Muroga formula does not apply, we must resort to numerics in the form of the iterative Blahut-Arimoto algorithm \cite{Blahut,Arimoto}. Moreover, even when the Muroga formula holds, a complete theoretical treatment is hindered by the intrinsic analytical complexity of the necessary calculations. 

For instance, the Muroga formula itself involves a matrix inverse, and the calculation of $t_{mix}(W)$---a prerequisite for calculating $\beta_{mix}$, $F_{mix}$, etc.---is essentially spectral in nature. Attempts to produce analytical approximations near $W = 11^T/n$ or $I$ are both stymied: in the first case, because $W$ is (close to being) noninvertible, and in the second case because it turns out to be necessary to go to second order in perturbation theory. 
\footnote{
If $W = I + \epsilon Q$, then in the notation and context of \S \ref{sec:basinsfollowcorners}, $\sum_k M_{jk} H_k = H_j + O(\epsilon^2)$: both sides equal $-\epsilon[Q_{jj}(1-\log \epsilon) + \sum_{\ell \ne j} Q_{j\ell} \log Q_{j\ell}]$ to first order in $\epsilon$.
}

With this in mind, our analysis of generic nonsingular constrained DMCs is necessarily incomplete, though we take pains to reach a local maximum of the ratio between results and effort here. In particular, we provide illustrative and robust numerical examples 
as well as qualitative analytical results below that illustrate the relationship between channel capacity and free energy.


In Figure \ref{fig:pardmc} we consider the partially constrained $3 \times 3$ channel matrix
\begin{equation}
\label{eq:exampleA}
W(u,v) := \left ( \begin{smallmatrix} 
W_{11} & (1-W_{11}) \cdot u & (1-W_{11}) \cdot (1-u) \\
(1-W_{22}) \cdot (1-v) & W_{22} & (1-W_{22}) \cdot v \\
W_{31} & W_{32} & W_{33}
\end{smallmatrix} \right ).
\end{equation}
Consequently, we get, e.g., $C(u,v) := C(W(u,v))$, etc. Note that this particular model has fixed probabilities for successful transmission of symbols.

Figure \ref{fig:pardmc} depicts the specific choice $W_{11} = .35$, $W_{22} = .05$, $W_{31} = .15$, $W_{32} = .2$, and $W_{33} = .65$. This example is in many ways generic. The most salient generic feature is that the basins of $F$ lie along ``corners'' of $C$, and that the basin width and corner sharpness are correlated. Together with similar behavior of $-1/\beta_{mix}$ in relation to $C$ (not shown), this gives a qualitative  explanation of why $\arg \min F_{mix}$ should be strongly correlated with $\arg \min C$: the basins of $F_{mix}$ and corners of $C$ both merge near each other and the respective minima.

\begin{figure}[htbp]
\centering
\includegraphics[trim = 60mm 125mm 60mm 115mm, clip, width=\columnwidth,keepaspectratio]{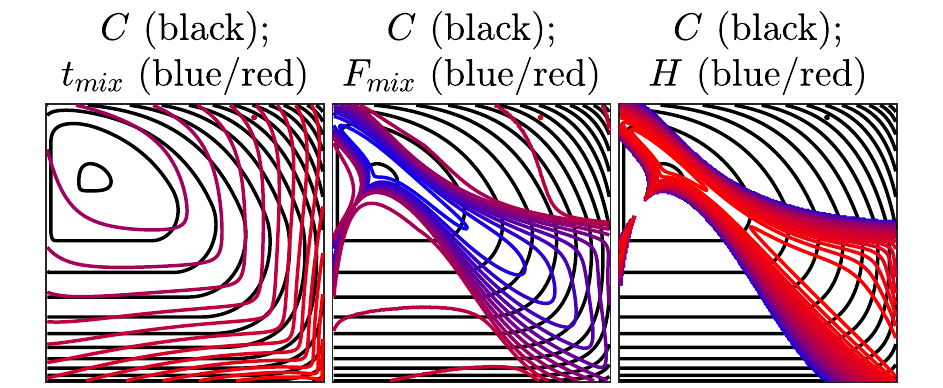}
\includegraphics[trim = 60mm 125mm 60mm 115mm, clip, width=\columnwidth,keepaspectratio]{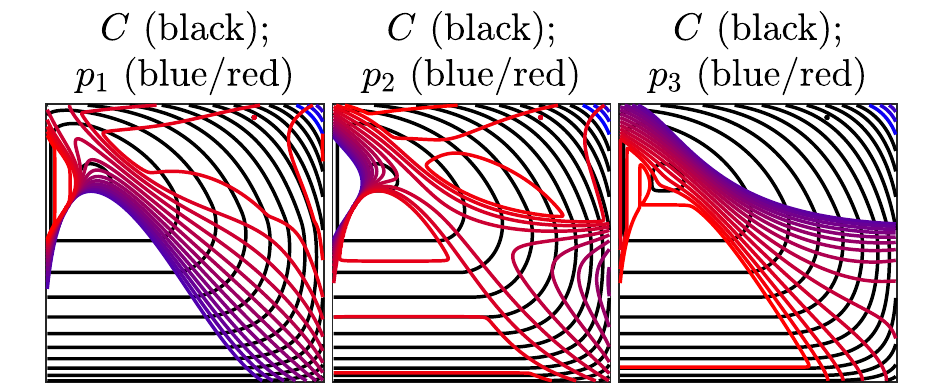}
\caption{ \label{fig:pardmc} (Top left) Contours of $C$ (black) and $t_{mix}$ (values increasing from {\color{blue}blue} to {\color{red}red}) for a DMC described in \eqref{eq:exampleA}; the axes are $(u,v)$ over the unit square. (Top center; right) As in the leftmost panel, but for $F_{mix}$ and $H$, respectively. (Bottom) As above, but for elements of the capacity-achieving distribution. Note that $p_1$, $p_2$, and $p_3$ are respectively zero in the lower left, upper left, and upper right regions. We have (see \S \ref{sec:basinsfollowcorners}) that $\partial C/\partial u = (\partial C/\partial W_{12})(\partial W_{12}/\partial u) + (\partial C/\partial W_{13})(\partial W_{13}/\partial u) = (1-W_{11}) p_1 (\psi^{(12)} - \psi^{(13)}) = 0$ where $p_1 = 0$, and similarly $\partial C/\partial v = (1-W_{22}) p_2 (-\psi^{(21)} + \psi^{(23)}) = 0$ where $p_2 = 0$. Together with the roughly capacity-emulating behavior of $t_{mix}$, this helps to explain why minima of capacity and free energy are correlated.}
\end{figure}

\subsection{\label{sec:basinsfollowcorners}Basins of $F$ follow corners of $C$}

We first recall the classical Muroga formula for the channel capacity of a DMC \cite{Muroga,Ash}. Let $W$ be an invertible channel matrix with $M := W^{-1}$ and define $H_j := -\sum_k W_{jk} \log W_{jk}$. Furthermore, define $d_j := \sum_i M_{ij} \exp \left ( - \sum_k M_{ik} H_k \right )$. The Muroga formula states that if $d > 0$, then $e^C = \sum_j \exp \left ( - \sum_k M_{jk} H_k \right )$, and $p := e^{-C} d$ is a capacity-achieving distribution. \footnote{As an aside, we note that the Muroga formula gives yet another indication of the similarity between the channel capacity and a free energy.}

\begin{theorem}
If the Muroga formula applies, then $\partial C/\partial W_{jk} = \psi^{(jk)} p_j$, where \begin{equation}
\psi^{(jk)} := \sum_m \left ( M_{km} - M_{jm} \right ) H_m + \log \frac{W_{jk}}{W_{jj}}
\end{equation} 
and $W_{jj} = 1 - \sum_{k \ne j} W_{jk}$ are dependent variables.
\end{theorem}

\begin{proof}
By the Muroga formula, we have for $j \ne k$ that
$$\frac{\partial}{\partial W_{jk}} e^C = \sum_\ell \left [ \exp \left ( - \sum_{m'} M_{\ell m'} H_{m'} \right ) \cdot \sum_m \phi_{\ell m}^{(jk)} \right ],$$
where
$$\phi_{\ell m}^{(jk)} := -\frac{\partial M_{\ell m}}{\partial W_{jk}} \cdot H_m - M_{\ell m} \cdot \frac{\partial H_m}{\partial W_{jk}}.$$.

For $A$ invertible, we have $0 = D_x \left ( AA^{-1} \right ) = D_x (A) \cdot A^{-1} + A \cdot D_x (A^{-1})$, so $D_x (A^{-1}) = -A^{-1} (D_x A) A^{-1}$. This yields that $\frac{\partial M_{\ell m}}{\partial W_{jk}} = -\sum_{a,b} M_{\ell a} \frac{\partial W_{ab}}{\partial W_{jk}} M_{b m}$. Meanwhile, a line of algebra shows that if $j \ne k$, then $\frac{\partial W_{ab}}{\partial W_{jk}} = \delta_{ja}(\delta_{kb} - \delta_{ab})$. This yields in turn that $\frac{\partial M_{\ell m}}{\partial W_{jk}} = -M_{\ell j} \left ( M_{km} - M_{jm} \right )$. 

Straightforward calculations also yield 
$$\frac{\partial}{\partial W_{jk}} \left ( W_{ma} \log W_{ma} \right ) = \delta_{jm}(\delta_{ka} - \delta_{ma}) \cdot (\log W_{ma} + 1)$$ and $\frac{\partial H_m}{\partial W_{jk}} = -\delta_{jm} \log \frac{W_{mk}}{W_{mm}}$. Collecting results, we obtain 
$$\phi_{\ell m}^{(jk)} = M_{\ell j} \left [ \left ( M_{km} - M_{jm} \right ) H_m + \delta_{jm} \log \frac{W_{mk}}{W_{mm}} \right ]$$
and $\sum_m \phi_{\ell m}^{(jk)} = M_{\ell j} \psi^{(jk)}$, so $\frac{\partial}{\partial W_{jk}} e^C = \psi^{(jk)} d_j$. Since $e^{-C} d = p$, the theorem follows.
\end{proof}

The key consequence of this result is that apart from an extremum, many partial derivatives of $C$ also vanish where components of the capacity-achieving distribution become zero. The Karush-Kuhn-Tucker conditions associated with the Muroga formula imply that if $p_j = 0$, then $\partial C/\partial W_{jk} = 0$ for all $k$. This means that (the graph and level surfaces of) $C$ will have (straight edges adjacent to) corners whose sharpness is dictated by the separation between regions where $p_j = 0$ and $p_k = 0$ for some $j \ne k$. 
\footnote{
As \cite{Ash} points out, if $W$ is invertible but $d \not > 0$, then the KKT conditions imply that a capacity-achieving distribution must have at least one zero component. A detailed analysis would begin by eliminating the appropriate rows of $W$ before dealing with the resulting nonsquare channel matrix in the manner of \cite{Cheng,Takano}. However, for our purposes it is sufficient to note that $\partial C/\partial W_{jk} = \psi^{(jk)} p_j$ still holds when $p \not > 0$. 
}
Meanwhile, as $p_j \downarrow 0$, we have that $\beta^{-1} \downarrow 0$ and $F \uparrow 0$. In other words, $F$ plateaus at its maximum possible value of $0$ wherever $p \not > 0$, and necessarily has basins between regions where $p_j = 0$ and $p_k = 0$ for some $j \ne k$. That is, \emph{basins of $F$ follow corners of $C$}.

This in turn explains why minima of $F$ occur near minima of $C$. $F$ is near a more constrained minimum at a corner of $C$, which is also where a more constrained minimum of $C$ occurs. As we let more entries of $W$ vary, the respective basins of $F$ coalesce near the progressively less constrained minima of $C$.

An analogous argument does not hold for entropy, as Figure \ref{fig:pardmc} makes clear by way of example. The distinguishing cause is the role played by $t_\infty$. The already-discussed relationship between the minimum of $t_\infty$ and $C$ acts in concert with the basins/corners principle (which on its own still applies to the entropy) to localize the minimum of $F$ near that of $C$.

\begin{figure}[htbp]
\centering
\includegraphics[trim = 60mm 125mm 60mm 115mm, clip, width=\columnwidth,keepaspectratio]{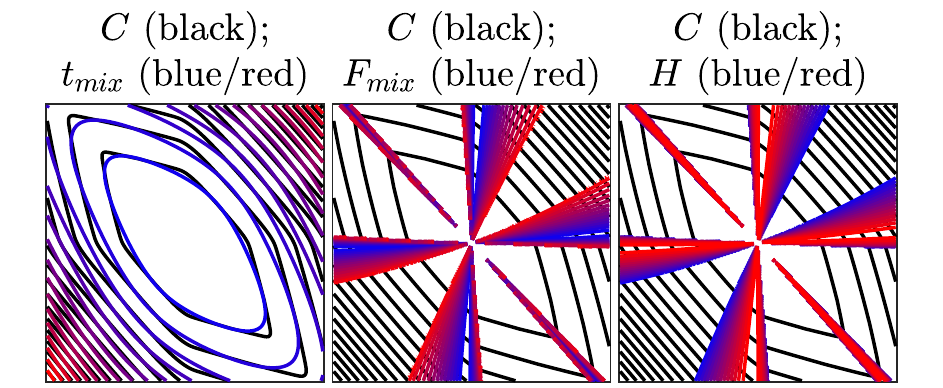}
\\
\includegraphics[trim = 60mm 125mm 60mm 115mm, clip, width=\columnwidth,keepaspectratio]{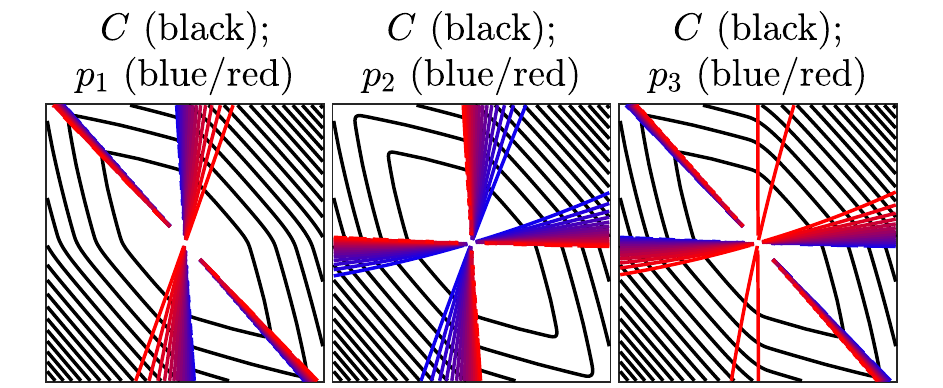}
\caption{ \label{fig:pardmc2} As in Figure \ref{fig:pardmc}, but for the example at the end of \S \ref{sec:basinsfollowcorners}.}
\end{figure}

Figure \ref{fig:pardmc2} provides a more striking example illustrating the general principle that basins of $F$ tend to follow corners of $C$ by considering convex combinations $W(u,v) = c^{(0)}\frac{11^T}{n} + c^{(u)}W^{(u)} + c^{(v)}W^{(v)} $ with $n = 3$, $c^{(0)} := 1- a(u+v-1)$, $c^{(u)} := a(u-\frac{1}{2})$, $c^{(v)} := a(v-\frac{1}{2})$, and with $a = 0.2$ and 
\begin{equation}
W^{(u)} = \left ( \begin{smallmatrix} 0.25 & 0.14 & 0.61 \\ 0.29 & 0.67 & 0.04 \\ 0.08 & 0.74 & 0.18 \end{smallmatrix} \right ); \quad W^{(v)} = \left ( \begin{smallmatrix} 0.31 & 0.04 & 0.65 \\ 0.50 & 0.10 & 0.40 \\ 0.01 & 0.58 & 0.41 \end{smallmatrix} \right ). \nonumber
\end{equation}
As mentioned above, this example also highlights the role of degeneracies: in order to have $C \approx \min C$, we must have $t_\infty \approx \min t_\infty$ as well as $F \approx \min F$ (note that these can be jointly satisfied).

\begin{theorem}
Near $\arg \min C$, we have $\psi^{(jk)} \approx 0$ for $j \ne k$; if also $\sum_\ell M_{j\ell}/p_\ell \approx \sum_\ell M_{k\ell}/p_\ell$, then $\partial \log Z/\partial W_{jk} \approx 0$. \end{theorem}

\begin{proof} For $j \ne k$, we have $\frac{\partial d_\ell}{\partial W_{jk}} = (M_{j\ell}-M_{k\ell})d_j + \psi^{(jk)} \sum_i M_{ij} M_{i\ell} \exp \left ( -\sum_m M_{im} H_m \right )$. Note that the partition function of the capacity-achieving distribution $p = e^{-C} d$ satisfies $\log Z = -\frac{1}{n} \sum_j \log p_j = C - \frac{1}{n} \sum_j \log d_j$, while as a consequence of the previous result we have that $\sum_\ell \frac{\partial \log d_\ell}{\partial W_{jk}} = d_j \sum_\ell \frac{M_{j\ell}-M_{k\ell}}{d_\ell} + \psi^{(jk)} \sum_i M_{ij} \exp \left ( -\sum_m M_{im} H_m \right ) \cdot \sum_\ell \frac{M_{i\ell}}{d_\ell}$. 

By the previous theorem, we have that $\frac{\partial \log Z}{\partial W_{jk}} = \psi^{(jk)} p_j - \frac{1}{n} \sum_\ell \frac{\partial \log d_\ell}{\partial W_{jk}}$
and for $j \ne k$ furthermore that $\sum_\ell \frac{M_{i\ell}}{p_\ell} \overset{\forall i}{\approx} n \Rightarrow \frac{\partial \log Z}{\partial W_{jk}} \approx 0$. The result follows.
\end{proof}

The assumptions of the theorem are satisfied in the ``good channel'' regime $W = I + \varepsilon Q$ where $Q$ is such that its rows sum to zero, its diagonal is strictly negative, and $\| \varepsilon Q \| \ll1$.
\footnote{
Note that in this regime that $t_{mix}$ becomes very large.
}
The theorem allows us to understand when and why $-\log Z = \beta F$ is approximately constant ($>1$) near $\arg \min C$. Since here $\beta^{-1} \approx F$ up to a constant factor, this turns out to inform how our results align with those of \cite{Reiss,ReissHuang} for discrete noiseless channels.

\section{\label{sec:closing}Remarks}

An important question is whether or not an \emph{exact} equality of the form $\arg \min C = \arg \min F$ holds for non-binary channels. Treatments of the binary case in this paper and in \cite{QianR}, along with general arguments for DMCs, strongly suggest that this ought to be the case. The alternative would imply a counterintuitive residual capacity for communication in thermodynamic equilibrium states that would vanish in a nonequilibrium state. 

With this in mind, taking $\arg \min C = \arg \min F$ as an \emph{Ansatz} provides strong constraints on $t_\infty$ that complement purely physical constraints discussed in \cite{Huntsman3}. This is significant because finding the ``correct'' form for $t_\infty$ (versus merely working with $t_{mix}$ as a proxy) is the central theoretical obstacle to applying the framework of effective statistical physics in a precise way to generic finite systems, and might inform nonequilibrium statistical physics as well.

\appendices

\section{\label{sec:mixingTime}$L^2$ mixing time}

We follow \cite{LPW} without comment or elaboration here. If $P$ is a (row-) stochastic matrix and $p$ the corresponding invariant distribution, then the Dirichlet form is $\mathcal{E}(f) = \frac{1}{2} \sum_{j,k}p_j P_{jk} (f_j - f_k)^2$ and the time reversal or $L^2$ adjoint of $P$ is given by $P^\dagger_{jk} := p_k P_{kj}/p_j$. It is easy to show that $pP^\dagger P = p$, and in turn that the Dirichlet form corresponding to $P^\dagger P$ satisfies $-\mathcal{E}_{P^\dagger P} = \mbox{Var}_p(Pf) - \mbox{Var}_p(f)$, where the variance of $f$ is $\mbox{Var}_{p}(f) := \sum_j p_j f_j^2 - (\sum_j p_j f_j)^2$. Define
\begin{equation}
\label{eq:lambdastar3}
\lambda_* := \inf_{f} \frac{\mathcal{E}_{P^\dagger P}(f)}{\mbox{Var}_p(f)},
\end{equation} 
where the infimum is over $f$ s.t. $\mbox{Var}_p(f) \ne 0$. If now we additionally define $S := \text{diag}(p) (P^\dagger P - I)$, then $\mathcal{E}(f) = -f^T S f$. With $T := (S + S^T)/2 = S$ and $U := -\text{diag}(p)^{-1/2} T \text{diag}(p)^{-1/2}$, we have
\begin{equation}
\label{eq:lambdastar4}
\lambda_* = \inf \left ( \mbox{spec}(U) \backslash \{ 0 \} \right ).
\end{equation}
$\lambda_*$ determines the $L^2$ convergence of the Markov process to stationarity; $t_{mix} := \lambda_*^{-1}$ is the $L^2$ mixing time. 


\section*{Acknowledgment} I am grateful to H. J. Gonzalez for pointing out a characterization of zero-capacity DMCs. 

%

\ifCLASSOPTIONcaptionsoff
  \newpage
\fi



%

\end{document}